\documentclass[runningheads]{llncs}
\usepackage{graphicx}
\usepackage{amsmath}
\usepackage{amssymb}
\usepackage{amsfonts}
\usepackage{amscd}
\usepackage{ stmaryrd } 

\usepackage[usenames]{color}
\usepackage{pgf,tikz}
\usetikzlibrary{arrows}
\usepackage{tikz-cd}  
\usepackage{hyperref}

\newcommand{\Exx}[2]{\mathbb{E}_{#1}\left(#2\right)}
\renewcommand{\Pr}[1]{\mathbb{P}\left(#1\right)}

\newcommand{\salg}[1]{\mathfrak {#1}}

\newcommand{\cf}{\chi}

\newcommand{\diff}[2]{\frac{\mathrm{d}  #1}{\mathrm{d}  #2}}

\def \d{\mbox{\(\,\mathrm{d}\)}}
 
\renewcommand{\epsilon}{\varepsilon}

\newcommand{\eq}[1]{\overset{\tiny{(#1)}}=}

\newcommand{\Rr}{\mathbb{R}}

\newcommand{\Nn}{\mathbb{N}}

\newcommand{\set}[2]{\{\,#1 \, \vert \, #2\,\} }
\newcommand{\bigset}[2]{\left\{\,#1 \, \big\vert \, #2\,\right\} }

\usepackage{mathtools} 

\newcommand{\ent}{S}

\begin{document}
\title{Entropy under disintegrations}
%
%
\author{Juan Pablo Vigneaux
\inst{1,2}
}
\authorrunning{J.~P. Vigneaux}
%
\institute{Institut de Math\'ematiques de Jussieu--Paris Rive Gauche (IMJ-PRG), Universit\'e de Paris, 8 place Aur\'elie N\'emours, 75013 Paris, France. \and 
Max Planck Institute for Mathematics in the Sciences, Inselstra{\ss}e 22, 04103 Leipzig, Germany. \\ 
\href{https://orcid.org/0000-0003-4696-4537}{orcid.org/0000-0003-4696-4537}
}
%
\maketitle              
\begin{abstract}
We consider the differential entropy of probability measures absolutely continuous with respect to a given $\sigma$-finite ``reference'' measure on an arbitrary measure space. We state the asymptotic equipartition property in this general case; the result is part of the folklore but our presentation is to some extent novel. Then we study a general framework under which such entropies satisfy a chain rule: disintegrations of measures. We give an asymptotic interpretation for conditional entropies in this case.  Finally, we apply our result to Haar measures in canonical relation. 

\end{abstract}

\section{Introduction}

It is part of the ``folklore'' of information theory that given any measurable space $(E,\salg B)$ with reference measure $\mu$, and a probability measure $\rho$ on $E$ that is absolutely continuous with respect to $\mu$ (i.e. $\rho \ll \mu$), one can define a  \emph{differential entropy} $S_{\mu}(\rho)= -\int_{E} \log( \diff{\rho}{\mu}) \d\rho$ that gives the exponential growth rate of the $\mu^{\otimes n}$-volume of a typical set of realizations of $\rho^{\otimes n}$. Things are rarely treated at this level of generality in the literature, so the first purpose of this article is to state the \emph{asymptotic equipartition property (AEP)} for $S_{\mu}(\rho)$. This constitutes a unified treatment of the discrete and euclidean cases, which shows (again) that the differential entropy introduced by Shannon is not an unjustified \emph{ad hoc} device as some still claim. 

Then we concentrate on a question that has been largely neglected: what is the most general framework in which one can make sense of the chain rule? This is at least possible for any disintegration of a measure.

\begin{definition}[Disintegration]\label{def:disintegration} Let $T:(E,\salg B)\to (E_T,\salg B_T)$ be a measurable map, $\nu$ a $\sigma$-finite measure on $(E,\salg B)$, and $\xi$ a $\sigma$-finite measure on $(E_T,\salg B_T)$. The measure $\nu$ has a disintegration $\{\nu_t\}_{t\in E_T}$ with respect to $T$ and $\xi$, or a $(T,\xi)$-disintegration, if
\begin{enumerate}
\item $\nu_t$ is a $\sigma$-finite measure on $\salg{B}$ concentrated on $\{T=t\}$, which means that $\nu_t(T\neq t)=0$ for $\xi$-almost every $t$;
\item for each measurable nonnegative function $f:E\to \Rr$,
\begin{enumerate}
\item $t\mapsto \int_E f \d \nu_t$ is measurable, 
\item $\int_E f\d \nu= \int_{E_T} \left(\int_E  f(x) \d\nu_t(x) \right)\d \xi(t)$.
\end{enumerate}
\end{enumerate}
\end{definition}

We shall see that if the reference measure $\mu$ has a $(T,\xi)$-disintegration $\{\mu_t\}_{t\in E_T}$, then any probability $\rho$ absolutely continuous with respect to it has a $(T,T_*\rho)$-disintegration; each $\rho_t$ is absolutely continuous with respect to $\mu_t$, and its density can be obtained normalizing the restriction of $\diff{\rho}{\mu}$ to $\{T=t\}$. Moreover, the following chain rule holds:
 \begin{equation}\label{eq:generalized_chain_rule}
 \ent_{\mu}(\rho) =  \ent_{\xi}(T_*\rho) + \int_{E_T} \ent_{\mu_t}(\rho_t)  \d T_*\rho(t).
 \end{equation}
 
We study the meaning of $\int_{E_T} \ent_{\mu_t}(\rho_t)  \d T_*\rho(t)$ in terms of asymptotic volumes. Finally, we show that our generalized chain rule can be applied to Haar measures in canonical relation.

\section{Generalized differential entropy}

\subsection{Definition and AEP} 
Let $(E_X, \salg B)$ be a measurable space, supposed to be the range of some random variable $X$, and let $\mu$ be a $\sigma$-finite measure $\mu$ on it. In applications, several examples appear: 
\begin{enumerate}
\item $E_X$ a countable set, $\salg B$ the corresponding atomic $\sigma$-algebra, and $\mu$ the counting measure;
\item $E_X$ euclidean space, $\salg B$ its Borel $\sigma$-algebra, and $\mu$ the Lebesgue measure;
\item More generally: $E_X$ a locally compact topological group, $\salg B$ its Borel $\sigma$-algebra, and $\mu$ some Haar measure;
\item $(E_X, \salg B)$ arbitrary and $\mu$ a probability measure on it, that might be a prior in a Bayesian setting or an initial state in a physical/PDE setting.
\end{enumerate}
The reference measure $\mu$ gives the relevant notion of volume.

Let $\rho$ is a probability measure on  $(E_X, \salg B)$ absolutely continuous with respect to $\mu$, and $f$ a representative of the Radon-Nikodym derivative $\diff{\rho}{\mu}\in L^1(E_X,\mu_X)$. The \emph{generalized differential entropy} of $\rho$ with respect to (w.r.t.) $\mu$ is defined as
\begin{equation}\label{eq:def_gen_entropy}
 \ent_\mu(\rho):= \Exx{\rho}{-\ln \diff{\rho}{\mu}} = -\int_{E_X} f(x) \log f(x) \d\mu(x).
\end{equation}
This was introduced by Csisz\'ar in \cite{Csiszar1973}, see also Eq. (8) in \cite{Koliander2016}. Remark that the set where $f=0$, hence $\log(f) = -\infty$, is $\rho$-negligible.

Let $\{X_i:(\Omega, \salg F, \mathbb P) \to (E_X, \salg B, \mu) \}_{i\in \Nn}$ be a collection of i.i.d random variables with law $\rho$. The density of the joint variable $(X_1,...,X_n) $ w.r.t. $\mu^{\otimes n}$ is given by  $f_{X_1,...,X_n} (x_1,...,x_n) = \prod_{i=1}^n f(x_i)$. If the Lebesgue integral  in \eqref{eq:def_gen_entropy}
is finite, then 
\begin{equation}
-\frac{1}{n} \log f_{X_1,...,X_n}(X_1,...,X_n) \to \ent_\mu(\rho)
\end{equation}
$\mathbb P$-almost surely (resp. in probability) as a consequence of the strong (resp. weak) law of large numbers. The convergence in probability is enough to establish the following result.

\begin{proposition}[Asymptotic Equipartition Property]\label{prop:AEP}
Let $(E_X,\salg B,\mu)$ be a $\sigma$-finite measure space, and $\rho$ a probability measure on $(E_X,\salg B)$ such that  $\rho \ll \mu$ and $\ent_\mu(\rho)$ is finite. For every $\delta >0$, set
\begin{equation*}
A_\delta^{(n)}(\rho;\mu)  := \bigset{(x_1,...,x_n)\in E_X^n }{\left|-\frac 1n  \log f_{X_1,...,X_n}(X_1,...,X_n) - \ent_\mu(\rho)\right| \leq \delta}.
\end{equation*}
Then,
\begin{enumerate}
\item\label{AEP1} for every $\epsilon>0$, there exists $n_0\in \Nn$ such that, for all $n\geq n_0$, $$\Pr{A_\delta^{(n)}(\rho;\mu)}>1-\epsilon;$$
\item\label{AEP2} for every $n\in \Nn$, $$\mu^{\otimes n}(A_\delta^{(n)}(\rho;\mu)) \leq \exp\{n(\ent_\mu(\rho)+\delta)\};$$
\item\label{AEP3} for every $\epsilon>0$, there exists $n_0\in \Nn$ such that, for all $n\geq n_0$, $$\mu^{\otimes n}(A_\delta^{(n)}(\rho;\mu)) \geq (1-\epsilon)  \exp\{n(\ent_\mu(\rho)-\delta)\}.$$
\end{enumerate}
\end{proposition}

We proved these claims in \cite[Ch.~12]{Vigneaux2019-thesis}; our proofs are very similar to the standard ones for (euclidean) differential entropy, see \cite[Ch.~8]{Cover2006}.

Below, we write $A_\delta^{(n)}$ if $\rho$ and $\mu$ are clear from context.

When $E_X$ is a countable set and $\mu$ the counting measure, every probability law $\rho$ on $E_X$ is  absolutely continuous with respect to $\mu$; if $p:E_X\to \Rr$ is its density, $\ent_{\mu}(\rho)$ corresponds to the familiar expression $-\sum_{x\in E_X} p(x)\log p(x)$. 

If $E_X= \Rr^n$, $\mu$ is the corresponding Lebesgue measure, and $\rho$ a probability law such that $\rho \ll \mu$, then  the derivative $\d \rho/\d\mu \in L^1(\Rr^n)$ corresponds to the elementary notion of density, and the quantity   $\ent_\mu(\rho)$ is the \emph{differential entropy} that was also introduced by Shannon in \cite{Shannon1948}. He remarked that the covariance of the differential entropy under diffeomorphisms is consistent with the measurement of randomness ``\emph{relative to an assumed standard}.'' For example, consider a linear automorphism of $\Rr^n$,  $\varphi(x_1,...,x_n)= (y_1,...,y_n)$,  represented by a matrix $A$. Set $\mu = \d x_1 \cdots \d x_n$ and $\nu = \d y_1 \cdots \d y_n$. It can be easily deduced from the change-of-variables formula that $\nu(\varphi(V)) = |\det A| \mu(V)$. Similarly,  $\varphi_*\rho$ has density $f(\varphi^{-1}(y))|\det A|^{-1}$ w.r.t. $\nu$, and this implies that $S_{\nu}(\varphi_*\rho) = S_{\mu}(\rho) + \log |\det A|$, cf. \cite[Eq.~8.71]{Cover2006}. Hence
\begin{equation}
\left| -\frac 1n \log \prod_{i=1}^n \diff{\varphi_*\rho}{\nu}(y_i) - S_{\nu}(\varphi_*\rho)\right| = \left| -\frac 1n \log \prod_{i=1}^n \diff{\rho}{\mu}(\varphi^{-1}(y_i)) - S_{\mu}(\rho)\right|,
\end{equation}
from which we deduce that $A_{\delta}^{(n)}(\varphi_* \rho, \nu) = \varphi^{\times n}(A_\delta^{(n)}(\rho;\mu))$ and consequently 
\begin{equation}
\nu^{\otimes n} (A_{\delta}^{(n)} (\varphi_*\rho; \nu)) = |\det A|^n \mu^{\otimes n}(A_{\delta}^{(n)}(\rho; \mu)),
\end{equation}
which is  consistent with the corresponding estimates given by Proposition \ref{prop:AEP}.

In the discrete case one could also work with any multiple of the counting measure, $\nu = \alpha \mu$, for $\alpha>0$. In this case, the chain rule for Radon-Nikodym derivatives (see  \cite[Sec.~19.40]{Hewitt1965}) gives
\begin{equation}
\diff{\rho}{\mu}= \diff{\rho}{\nu}\diff{\nu}{\mu} = \alpha \diff{\rho}{\nu},
\end{equation}
and therefore $
S_{\mu}(\rho) =S_{\nu}(\rho) - \log\alpha.$
Hence the discrete entropy depends on the choice of reference measure, contrary to what is usually stated. This function is invariant under a bijection of finite sets, but taking on both sides the counting measure as reference measure. The proper analogue of this in the euclidean case is a measure-preserving transformation (e.g. $|\det A|=1$ above), under which the differential entropy \emph{is} invariant.

For any $E_X$, if $\mu$ is  a probability law, the expression $\ent_\mu(\rho)$ is the opposite of the \emph{Kullback-Leibler divergence} $D_{KL}(\rho||\mu) := -S_{\mu}(\rho).$ The positivity of the divergence follows from a customary application of Jensen's inequality or from the asymptotic argument given in the next subsection.

The asymptotic relationship between volume and entropy given by the AEP can be summarized as follows:

\begin{corollary}\label{cor:entropy_and_volume}
$$\lim_{\delta \to 0} \lim_{n\to \infty} \frac{1}{n}\log \mu^{\otimes n}(A_\delta^{(n)}(\rho;\mu))=S_{\mu}(\rho).$$
\end{corollary}


\subsection{Certainty, positivity and divergence}\label{sec:certainty_divergence}

Proposition \ref{prop:AEP} gives a meaning to the divergence and the positivity/negativity of $S_{\mu}(\rho)$.
\begin{enumerate}
\item Discrete case: let $E_X$ be a countable set and $\mu$ be the counting measure.  Irrespective of $\rho$, the cardinality of $\mu^{\otimes n}(A_\delta^{(n)}(\rho;\mu))$ is at least $1$, hence the limit in Corollary \ref{cor:entropy_and_volume} is always positive, which establishes $\ent_\mu(\rho)\geq 0$. The case $S_\mu(\rho)= 0$ corresponds to certainty: if $\rho=\delta_{x_0}$, for certain $x_0\in E_X$, then $A_\delta^{(n)}=\{(x_0,...,x_0)\}$. 

\item Euclidean case: $E_X$ Euclidean space, $\mu$ Lebesgue measure. The differential entropy is negative if the volume of the typical set is (asymptotically) smaller than $1$.  Moreover, the divergence of the differential entropy to $-\infty$ correspond to asymptotic concentration on a $\mu$-negligible set. For instance, if $\rho$ has  $\mu(B(x_0,\epsilon))^{-1} \cf_{B(x_0,\epsilon)}$, then $\ent_{\lambda_d}(\rho) = \log(|B(x_0,\epsilon)|) = \log(c_d\epsilon^d)$, where $c_d$ is a constant characteristic of each dimension $d$.
By part \eqref{AEP2} of Proposition \ref{prop:AEP}, $
|A_\delta^{(n)}| \leq \exp(nd\log \epsilon + Cn),$
which means that, for fixed $n$, the volume goes to zero as $\epsilon\to 0$, as intuition would suggest. Therefore, \emph{ the divergent entropy is necessary to obtain the good volume estimates.}

\item Whereas the positivity of the (discrete) entropy arises from a \emph{lower bound} to the volume of typical sets, the positivity of the Kullback-Leibler is of a different nature: it comes from an \emph{upper bound}. In fact, when $\mu$ and $\rho$ are probability measures such that $\rho \ll \mu$, the inequality $\mu^{\otimes n }(A_\delta^{(n)}(\rho;\mu))\leq 1$ holds for any $\delta$, which translates into $S_\mu(\rho)\leq 0$ and therefore $D_{KL}(\rho||\mu)\geq 0$. In general, there is no upper bound for the divergence.
\end{enumerate}

Remark that entropy maximization problems are well defined when the reference measure is a finite measure. 

\section{Chain rule}

\subsection{Disintegration of measures }

We summarize in this section some fundamental results on disintegrations as presented in \cite{Chang1997}. Throughout it,  $(E,\salg B)$ and $(E_T,\salg B_T)$ are measurable spaces equipped with $\sigma$-finite measures $\nu$ and $\xi$, respectively, and $T:(E,\salg B)\to (E_T,\salg B_T)$ is a measurable map.

Definition \ref{def:disintegration} is partly motivated by the following observation: when $E_T$ is finite and $\salg B_T$ is its algebra of subsets $2^{E_T}$, we can associate to any probability $P$ on $(E,\salg B)$ a $(T, T_*P)$-disintegration given by the conditional measures $P_t:\salg B \to \Rr, B\mapsto P(B\cap \{T=t\})/P(T=t)$, indexed by $t\in E_T$. In particular, 
\begin{equation}\label{eq:disintegration-discrete}
P(B) = \sum_{t\in E_T} P(T=t) P_t(B).
\end{equation}
 Remark that $P_t$ is only well defined on the maximal set of $t\in E_T$ such that $T_*P(t) > 0$, but only these $t$ play a role in the disintegration  \eqref{eq:disintegration-discrete}.

General disintegrations give \emph{regular versions} of conditional expectations. Let  $\nu$ be a probability measure, $\xi = T_*\nu$, and $\{\nu_t\}$ the corresponding $(T,\xi)$-disintegration. Then the function $x\in E\mapsto \int_E \chi_B(x) \d \nu_{T(x)}$---where $\chi_B$ denotes the characteristic function---is $\sigma_T$ measurable and a regular version of the conditional probability $\nu_{\sigma(T)}(B)$ as defined by Kolmogorov.


Disintegrations exist under very general hypotheses. For instance, if $\nu$ is Radon, $T_*\nu \ll \xi$, and $\salg B_T$ is countably generated and contains all the singletons $\{t\}$, then $\nu$ has a $(T,\xi)$-disintegration. The resulting measures $\nu_t$ measures are uniquely determined up to an almost sure equivalence. See \cite[Thm.~1]{Chang1997}.

As we explained in the introduction, a disintegration of a reference measure induces disintegrations of all measures absolutely continuous with respect to it. 

\begin{proposition}\label{prop:properties-disintegrations-with-densities}
Let $\nu$ have a $(T,\xi)$-disintegration $\{\nu_t\}$ and let $\rho$ be absolutely continuous with respect to $\nu$ with finite density $r(x)$, with each $\nu$, $\xi$ and $\rho$ $\sigma$-finite.
\begin{enumerate}
\item\label{properties-disintegrations-with-densities-1} The measure $\rho$ has a $(T,\xi)$-disintegration $\{\tilde \rho_t\}$ where each $\tilde \rho_t$ is dominated by the corresponding $\nu_t$, with density $r(x)$.
\item\label{properties-disintegrations-with-densities-2} The image measure $T_*\rho$ is absolutely continuous with respect to $\xi$, with density $\int_E  r \d \nu_t$.
\item\label{properties-disintegrations-with-densities-3} The measures $\{\tilde \rho_t\}$ are finite for $\xi$-almost all $t$ if and only if $T_*\rho$ is $\sigma$-finite.
\item\label{properties-disintegrations-with-densities-4} The measures $\{\tilde \rho_t\}$ are probabilities for $\xi$-almost all $t$ if and only if $\xi=T_*\rho$.
\item\label{properties-disintegrations-with-densities-5} If $T_*\rho$ is  $\sigma$-finite then $0<\nu_t r < \infty$ $T_*\nu$-almost surely, and the measures $\{\rho_t\}$ given by
$$ \int_E f \d \rho_t = \frac{\int_E fr \d \nu_t}{\int_E r \d  \nu_t}$$
are probabilities that give a $(T, T_*\rho)$-disintegration of $\rho$.
\end{enumerate}
\end{proposition}

\begin{example}[Product spaces] \label{ex:disintegrations_product-spaces}
We suppose that $(E, \salg B, \nu)$ is the product of two measured spaces spaces $(E_T, \salg B_T, \xi)$ and $(E_S, \salg B_S, \nu)$, with $\xi$ and $\nu$ both $\sigma$-finite. Let $\nu_t$ be the image of $\nu$ under the inclusion $s\mapsto (t,s)$. Then Fubini's theorem implies that $\nu_t$ is a $(T,\xi)$-disintegration of $\nu$. (Remark that  $\xi \neq T_*\nu$. In general, the measure $T_*\nu$ is not even $\sigma$-finite.) If $r(t,s)$ is the density of a probability $\rho$ on $(E,\salg B)$, then $\rho_t \ll \nu_t$ with density $r(t,s)$---the value of $t$ being fixed---and $\tilde \rho_t$ is a probability supported on $\{T=t\}$ with density
$ r(t,s)/\int_{E_S} r(t,s)\d \nu( s).$
\end{example}

\subsection{Chain rule under disintegrations}

Any disintegration gives a chain rule for entropy.

\begin{proposition}[Chain rule for general disintegrations]\label{prop:chain-rule-disintegration}
 Let $T:(E_X,\salg B_X) \to (E_Y,\salg B_Y)$ be a measurable map between arbitrary measurable spaces, $\mu$ (respectively  $\nu$) a $\sigma$-finite measure on $(E_X,\salg B_X)$ (resp. $(E_Y,\salg B_Y)$), and $\{\mu_y\}$ a $(T, \nu)$-disintegration of $\mu$. Then any probability measure $\rho$ absolutely continuous w.r.t. $\mu$, with density $r$, has a $(T, \nu)$-disintegration $\{\tilde \rho_y\}_{y\in Y}$ such that  for each $y$, $\tilde \rho_y = r \cdot \mu_y$. Additionally, $\rho$ has a $(T,T_*\rho)$-disintegration $\{ \rho_y\}_{y\in Y}$ such that each $\rho_y$ is a probability measure with density $r/\int_{E_X} r \d \mu_y $ w.r.t. $\mu_y$, and the following chain rule holds:
 \begin{equation}\label{eq:generalized_chain_rule}
 \ent_{\mu}(\rho) =  \ent_{\nu}(T_*\rho) +  \int_{E_Y} \ent_{\mu_y}(\rho_y)\d T_*\rho(y).
 \end{equation}
 \end{proposition}
 \begin{proof}
For convenience, we use here linear-functional notation:  $ \int_X f(x)\d\mu(x)$ is denoted $\mu(f)$ or $\mu^x(f(x))$ if we want to emphasize the variable integrated. 

Almost everything is a restatement of Proposition \ref{prop:properties-disintegrations-with-densities}. Remark that  $y\mapsto \mu_y(r)$ is the density of $T_*\rho$ with respect to $\nu$. 
 
Equation \eqref{eq:generalized_chain_rule} is established as follows:
 \begin{align}
 \ent_{\mu}(\rho) &\eq{def} \rho\left(-\log \diff{\rho}{\mu}\right)= T_*\rho^y \left(\rho_y \left(-\log \diff{\rho}{\mu}\right)\right) \label{eq:chain_rule_proof_1}\\
 &= T_*\rho^y \left(\rho_y \left(-\log \diff{\rho_y}{\mu_y} - \log \mu_y(r)\right)\right) \label{eq:chain_rule_proof_2}\\
 &=  T_*\rho^y \left(\rho_y \left(-\log \diff{\rho_y}{\mu_y}\right)\right)  + T_*\rho^y \left(-\log \diff{T_*\rho}{\nu}\right) \\
  &=  T_*\rho^y \left(\ent_{\mu_y}(\rho_y)\right)  + \ent_{\nu}(T_*\rho),\label{chain_rule_dis}
 \end{align}
 where  \eqref{eq:chain_rule_proof_1} is the fundamental property of the $T$-disintegration $\{\rho_y\}_y$ and \eqref{eq:chain_rule_proof_2} is justified by the equalities $$\diff{\rho}{\mu} = \diff{\tilde \rho_y}{\mu_y} =m_y(r)  \diff{ \rho_y}{\mu_y}.$$
 \end{proof}

\begin{example}\label{ex:product-case-continued} From the computations of Example \ref{ex:disintegrations_product-spaces}, it is easy to see that if $E_X=\Rr^{n}\times \Rr^m$, $\mu$ is the Lebesgue measure, and $T$ is the projection on the $\Rr^n$ factor, then \eqref{eq:generalized_chain_rule} corresponds to the familiar chain rule for Shannon's differential entropy. 
\end{example}
 \begin{example}[Chain rule in polar coordinates]
Let $E_X=\Rr^2\setminus\{0\}$, $\mu$ be the Lebesgue measure $\d x\d y$ on $\Rr^2$, and $\rho = f \d x\d y$ a probability measure. Every point $\vec{v}\in E_X$ can be parametrized by cartesian coordinates $(x,y)$ or polar coordinates $(r,\theta)$, i.e. $\vec{v}=\vec{v}(x,y) = \vec{v}(r,\theta)$. The parameter $r$ takes values from the set $E_R=]0,\infty[$, and $\theta$ from $E_{\Theta}=[0,2\pi[$; the functions $R:E_X\to E_R,\;v \mapsto r(\vec{v})$ and $\Theta:E_X\to E_{\Theta},\: \vec{v}\mapsto \theta(\vec v)$ can be seen as random variables with laws $R_*\rho$ and $\Theta_*\rho$, respectively. We equip $E_R$ (resp. $E_{\Theta}$) with the Lebesgue measure $\mu_R=\d r$ (resp. $\mu_H= \d\theta$).

The measure $\mu$ has a $(R,\mu_R)$-disintegration $\{ r \d\theta\}_{r\in E_R}$; here $r\d\theta$  is the uniform measure on $R^{-1}(r)$ of total mass $2\pi r$. This is a consequence of the change-of-variables formula:
\begin{equation}\label{eq:change-var-polar}
\int_{\Rr^2} \varphi(x,y)\d x \d y = \int_{[0,\infty[} \left(\int_0^{2\pi} \varphi(r,\theta) r \d \theta\right) \d r,
\end{equation}
which is precisely the disintegration property. Hence, according to Proposition \ref{prop:properties-disintegrations-with-densities}, $\rho$ disintegrates into probability measures $\{\rho_r\}_{r\in E_R}$, with each $\rho_r$ concentrated on $\{R=r\}$, absolutely continuous w.r.t. $\mu_r=r\d \theta$ and with density $f/\int_{0}^{2\pi} f(r,\theta)r\d\theta$. The exact chain rule \eqref{eq:generalized_chain_rule} holds in this case.
 
 This should be compared with Lemma 6.16 in  \cite{Lapidoth2003capacity}. They consider the random vector $(R,\Theta)$ as an $\Rr^2$ valued random variable, and the reference measure to be $\nu=\d r\d\theta$. The change-of-variables formula implies that $(R,\Theta)$ has density $rf(r,\theta)$ with respect to $\nu$, so 
 $S_{\nu}(\rho) = S_{\mu}(\rho) -\mathbb E_{\rho}(\log R)$. Then they apply the standard chain rule to $S_{\nu}(\rho)$, i.e. as in Examples \ref{ex:disintegrations_product-spaces} and  \ref{ex:product-case-continued}, to obtain a deformed chain rule for $S_{\mu}(\rho)$:
 \begin{equation}
 S_{\mu}(\rho) = S_{\mu_R} (R_*\rho) + \int_0^{\infty} \left(-\int_0^{2\pi} \log\left(\frac{f}{\int_{0}^{2\pi} f\d\theta} \right) \frac{f\d\theta}{\int_{0}^{2\pi} f\d\theta} \right) +  \mathbb E_{\rho}(\log R).
 \end{equation}
Our term $\int_{E_R} \ent_{\mu_r}(\rho_r)\d R_*\rho(r)$ comprises the last two terms in the previous equation. 
\end{example}

\begin{remark} Formula \eqref{eq:change-var-polar} is a particular case of the \emph{coarea formula} \cite[Thm.~2.93]{Ambrosio2000}, which gives a disintegration of the Hausdorff measure $\mathcal H^N$ restricted to a countably  $\mathcal H^N$-rectifiable subset $E$ of $\Rr^M$ with respect to a Lipschitz map $f:\Rr^M\to \Rr^k$ (with $k\leq N$) and the Lebesgue measure on $\Rr^k$. So the argument of the previous example also applies to the extra term $-\mathbb E_{(\mathbf x, \mathbf y)}[\log J_{\mathfrak p_{\mathbf y}}^{\mathcal E} (\mathbf x, \mathbf y)]$ in the chain rule of \cite[Thm.~41]{Koliander2016}, which could be avoided by an adequate choice of reference measures.
\end{remark}

Combining Corollary \ref{cor:entropy_and_volume} and the preceding proposition, we get a precise interpretation of the conditional term in terms of asymptotic growth of the volume of \emph{slices} of the typical set.

\begin{proposition}Keeping the setting of the previous proposition,
$$\lim_{\delta \to 0} \lim_{n\to \infty} 
\frac{1}{n} \log \left(\frac{\int_{E_Y} \mu^{\otimes n}_y(A_\delta^{(n)}(\rho;\mu)) \d \nu^{\otimes n} (y) }{\nu^{\otimes n}(A^{(n)}_\delta(T_*\rho;\nu))}\right) = \int_{E_Y} \ent_{\mu_y}(\rho_y)\d T_*\rho(y).$$
\end{proposition}
\begin{proof}
It is easy to prove that if $\{\mu_y\}_y$ is a $(T,\nu)$-disintegration of $\mu$, then $\{\mu_y^{\otimes n}\}_y$ is a $(T^{\times n},\nu^{\otimes n})$-disintegration of $\mu^{\otimes n}$. The disintegration property reads
\begin{equation}
\mu^{\otimes n}(A) = \int_{E_Y} \mu^{\otimes n}_y(A) \d \nu^{\otimes n}(y),
\end{equation}
for any measurable set $A$. Hence
\begin{equation}
\log \mu^{\otimes n}(A_\delta^{(n)}(\rho;\mu)) = \log \nu^{\otimes n}(A_\delta^{(n)}(T_*\rho;\nu)) + \log \frac{\int_{E_Y} \mu_y^{\otimes n}(A_\delta^{(n)}(\rho;\mu)) \d \nu^{\otimes n} (y)}{\nu^{\otimes n}(A_\delta^{(n)}(T_*\rho;\nu))}.
\end{equation}
The results follows from the application of $\lim_{\delta\to 0} \lim_n \frac{1}{n}$ to this equality and comparison of the result with the chain rule.
\end{proof}
In connection to this result, remark that $(T_*\rho)^{\otimes n}$ concentrates on $A_\delta^{(n)}(T_*\rho;\nu)$ and has approximately density $1/\nu^{\otimes n}(A_\delta^{(n)}(T_*\rho;\nu))$, so $$\int_{E_Y} \mu^{\otimes n}_y(A_\delta^{(n)}(\rho;\mu)) \frac{1}{\nu^{\otimes n}(A^{(n)}_\delta(T_*\rho;\nu))}\d \nu^{\otimes n} (y)$$ is close to an average of $\mu^{\otimes n}_y(A_\delta^{(n)}(\rho;\mu)\cap T^{-1}(y))$, the ``typical part'' of each fiber $T^{-1}(y)$, according to the ``true'' law $(T_*\rho)^{\otimes n}$. 

\subsection{Locally compact topological groups}

Given a locally compact topological group $G$, there is a unique left-invariant positive measure (left Haar measure) up to a multiplicative constant \cite[Thms. 9.2.2 \& 9.2.6]{Cohn2013}. A particular choice of left Haar measure will be denoted by $\lambda$ with superscript $G$ e.g. $\lambda^G$. The disintegration of Haar measures is given by Weil's formula.

\begin{proposition}[Weil's formula]\label{prop:Weil_formula}
Let $G$ be a locally compact group and $H$ a closed normal subgroup of $G$. Given  Haar measures on two groups among $G$, $H$ and $G/H$, there is a Haar measure on the third one such that, 
for any integrable function $f:G\rightarrow \mathbb{R}$,
\begin{equation}
\int_G f(x) d\lambda^G(x)=\int_{G/H} \left( \int_H f(xy)\d\lambda^H(y)\right) \d\lambda^{G/H}(xH).
\end{equation}
\end{proposition}
The three measures are said to be in \emph{canonical relation}, which is written $\lambda^G=\lambda^{G/H}\lambda^H$. For a proof of Proposition \ref{prop:Weil_formula}, see  pp. 87-88 and Theorem 3.4.6 of \cite{Reiter2000}.

For any element $[g]$ of $G/H$, representing a left coset $gH$, let us denote by $\lambda^H_{[g]}$ the image of $\lambda^H$ under the inclusion $\iota_g: H\to G, \:h\mapsto gH$. This is well defined i.e. does not depend on the chosen representative $g$: the image of $\iota_g$ depends only on the coset $gH$, and if $g_1,g_2$ are two elements of $G$ such that $g_1H = g_2H$, and $A$ is subset of $G$, the translation $h\mapsto g_2^{-1} g_1 h$ establishes a bijection   $\iota_{g_1}^{-1}(A)\overset\sim\to \iota_{g_2}^{-1}(A)$; the left invariance of the Haar measure implies that $\lambda^H(\iota_{g_1}^{-1}(A)) = \lambda^H(\iota_{g_2}^{-1}(A))$ i.e. $(\iota_{g_1})_* \lambda^H = (\iota_{g_2})_* \lambda^H$ as claimed. Proposition \ref{prop:Weil_formula} shows then that $\{\lambda^H_{[g]}\}_{[g]\in G/H}$ is a $(T,\lambda^{G/H})$-disintegration of $\lambda^G$. In view of this and Proposition \ref{prop:chain-rule-disintegration}, the following result follows.

\begin{proposition}[Chain rule, Haar case]
Let $G$ be a locally compact group, $H$ a closed normal subgroup of $G$, and $\lambda^G$, $\lambda^H$, and $\lambda^{G/H}$  Haar measures in canonical relation. Let $\rho$ be a probability measure on $G$. Denote by $T:G\to G/H$ the canonical projection. Then, there is $T$-disintegration $\{\rho_{[g]}\}_{[g]\in G/H}$ of $\rho$ such that each $\rho_{[g]}$ is a probability measure, and
\begin{equation}\label{eq:chain-rule-haar}
\ent_{\lambda^G}(\rho) = \ent_{\lambda^{G/H}} (\pi_*\rho) + \int_{G/H} \ent_{\lambda_{[g]}^H} (\rho_{[g]}) \d \pi_*\rho([g]).
\end{equation}
\end{proposition}

\bibliographystyle{splncs04}
\bibliography{../Bibliography}

\end{document}